%% file: main-arxiv-2.tex
\documentclass[a4paper,11pt]{article}
\usepackage[T1]{fontenc}
\usepackage{lmodern}

\usepackage{fullpage}

\usepackage[utf8]{inputenc} 
\usepackage[T1]{fontenc}    
\usepackage{hyperref}       
\usepackage{url}            
\usepackage{booktabs}       
\usepackage{amsfonts}       
\usepackage{nicefrac}       
\usepackage{microtype}      
\usepackage{lipsum}		
\usepackage{graphicx}
\usepackage{natbib}
\usepackage{doi}
\usepackage{mathtools}
\usepackage{amsthm}
\usepackage{xcolor}

\usepackage[ruled]{algorithm2e} 
\usepackage{bbm}

\newtheorem{example}{Example}[section]
\newtheorem{remark}{Remark}[section]

\newtheorem{definition}{Definition}[section]

\newtheorem{theorem}{Theorem}[section]
\newtheorem{lemma}[theorem]{Lemma}
\newtheorem{claim}[theorem]{Claim}

\newcommand{\Omit}[1]{}

\definecolor{WildStrawberry}{RGB}{255,67,164}
\input{notations}

\hypersetup{
pdftitle={On the Significance of Knowing the Arrival Order in Bayesian Online Settings},
pdfsubject={q-bio.NC, q-bio.QM},
pdfauthor={Tomer Ezra, Michal Feldman, Nick Gravin, Zhihao Gavin Tang}
,pdfkeywords={Prophet Inequality, Online Decision Making, Optimal Stopping},
}

\begin{document}

\title{``Who is Next in Line?''\\
On the Significance of Knowing the Arrival Order in\\ Bayesian Online Settings\thanks{This work is supported by Science and Technology Innovation 2030 –“New Generation of Artificial Intelligence” Major Project No.(2018AAA0100903), Innovation Program of Shanghai Municipal Education Commission, Program for Innovative Research Team of Shanghai University of Finance and Economics (IRTSHUFE) and the Fundamental Research Funds for the Central Universities. This project has received funding from the European Research Council (ERC) under the European Union's Horizon 2020 research and innovation program (grant agreement No. 866132), by the Israel Science Foundation (grant number 317/17), by an Amazon Research Award, and by the NSF-BSF (grant number 2020788). Tomer Ezra  was partially supported by the ERC Advanced Grant 788893 AMDROMA ``Algorithmic and Mechanism Design Research in Online Markets'' and MIUR PRIN project ALGADIMAR ``Algorithms, Games, and Digital Markets''. Zhihao Gavin Tang is supported by NSFC grant 61902233. Nick Gravin is supported by NSFC grant 62150610500.
} 
}


\author{
Tomer Ezra\thanks{Sapienza University of Rome, Email:  \texttt{ezra@diag.uniroma1.it}}
\and
Michal Feldman\thanks{Tel Aviv University, Email:  \texttt{mfeldman@tau.ac.il}}
\and
Nick Gravin\thanks{ITCS, Shanghai University of Finance and Economics, Email: \texttt{\{nikolai, tang.zhihao\}@mail.shufe.edu.cn}}
\and
Zhihao Gavin Tang\footnotemark[3]
}
\date{}

\maketitle
\thispagestyle{empty}

\begin{abstract}
	\input{abstract}
\end{abstract}

\clearpage

\setcounter{page}{1}


\section{Introduction}
\label{sec:intro}

\input{intro-new}

\section{Model and Preliminaries}
\label{sec:prelim}
\input{prelim}

\subsection{Single-Threshold Algorithms}
\label{sec:prelim-st}
\input{prelim2}

\section{Maximizing the Expected Value}
\label{sec:expectation}
\input{max_expectation}

\subsection{Single-Threshold Algorithms}
\label{sec:expectation-st}
\input{max_expectation_st}

\section{Maximizing the Probability of Catching the Maximum Value}
\label{sec:prob}
\input{max_prob}

\subsection{Single-Threshold Algorithms}
\label{sec:prob-st}
\input{max_prob_st}

\input{open_problems}


\bibliographystyle{abbrvnat}
\bibliography{order}

\end{document}

%% file: notations.tex
\newcommand{\be}{\begin{equation}}
\newcommand{\ee}{\end{equation}}
\newcommand{\beq}{\begin{equation*}}
\newcommand{\eeq}{\end{equation*}}

\newcommand{\eps}{\varepsilon}

\newcommand{\vmax}{v_\text{max}}

\newcommand{\AutoAdjust}[3]{\mathchoice{ \left #1 #2  \right #3}{#1 #2 #3}{#1 #2 #3}{#1 #2 #3} }
\newcommand{\Xcomment}[1]{{}}

\newcommand{\InBrackets}[1]{\AutoAdjust{[}{#1}{]}}
\newcommand{\Ex}[2][]{\operatorname{\mathbf E}_{#1}\InBrackets{#2}}

\newcommand{\Prx}[2][]{\operatorname{\mathbf{Pr}}_{#1}\InBrackets{#2}}

\newcommand{\eqdef}{\overset{\mathrm{def}}{=\mathrel{\mkern-3mu}=}}
\newcommand{\vect}[1]{\ensuremath{\mathbf{#1}}}

\newcommand{\RN}[1]{%
  \textup{\uppercase\expandafter{\romannumeral#1}}%
}

\newcommand\restr[2]{{
  \left.\kern-\nulldelimiterspace 
  #1 
  \vphantom{\big|} 
  \right|_{#2} 
  }}
\def\prob{\Prx}

\newcommand{\alg}{\textsf{ALG}}
\newcommand{\opt}{\textsf{OPT}}

\newcommand{\idr}[1]{{\mathbbm{1}({#1})}}

\newcommand{\vals}{\vec{v}}

\newcommand{\dist}{\mathbf{F}}
\newcommand{\dists}{\vect{\dist}}


%% file: abstract.tex
We introduce a new measure for the performance of online algorithms in Bayesian settings, where the input is drawn from a known prior, but the realizations are revealed one-by-one in an online fashion.
Our new measure is called {\em order-competitive ratio}. 
It is defined  as the worst case (over all distribution sequences) ratio between the performance of the best {\em order-unaware} and {\em order-aware} algorithms, and quantifies the loss that is incurred due to lack of knowledge of the arrival order. 
Despite the growing interest in the role of the arrival order on the performance of online algorithms, this loss has been overlooked thus far.

We study the order-competitive ratio in the paradigmatic {\em prophet inequality} problem, for the two common objective functions of (i) maximizing the expected value, and (ii) maximizing the probability of obtaining the largest value; and with respect to two families of algorithms, namely (i) adaptive algorithms, and (ii) single-threshold algorithms.
We provide tight bounds for all four combinations,  with respect to deterministic algorithms.
Our analysis requires new ideas and departs from standard techniques. 
In particular, our adaptive algorithms inevitably go beyond single-threshold algorithms.
The results with respect to the order-competitive ratio measure capture the intuition that adaptive algorithms are stronger than single-threshold ones, and may lead to a better algorithmic advice than the classical competitive ratio measure.



%% file: intro-new.tex
As part of the growing literature on ``beyond worst case analysis" (see \citep{roughgarden2021beyond} for a textbook treatment), the area of online algorithms has seen a notable shift from adversarial analysis to average-case analysis. 
A prominent framework is {\em Bayesian online settings}, where the input is assumed to be drawn from a known prior, but the actual realizations are revealed in an online fashion.
Many online problems have been studied within the Bayesian online framework, such as prophet inequality, stochastic matching, Steiner tree, and welfare and revenue maximization in combinatorial auctions \cite{krengel1978semiamarts, krengel1977semiamarts,samuel1984comparison,GravinW19,EzraFGT20,DBLP:conf/soda/GargGLS08,FeldmanGL15}.

The most common performance measure for the analysis of online {algorithms} is the {\em competitive ratio}, defined as the ratio between the performance of the algorithm and the optimal offline solution. That is, the performance of an online algorithm is evaluated by its performance relative to the performance of a ``prophet” who can see into the future.
The competitive ratio measure has been applied directly to the Bayesian online setting, where the performance of both the algorithm and the optimal solution is taken in expectation over the input.

However, the prophet benchmark may be too strong for Bayesian online settings.
In particular, it may provide an overly pessimistic performance prediction, and perhaps more concerning, it may fail to differentiate between different classes of algorithms, and may consequently provide  misleading algorithmic advice. 
A natural way to address this problem is to propose a more realistic benchmark. 
Indeed, there has been much interest recently in considering alternative, more realistic, benchmarks for this setting (see Section~\ref{sec:related}). 

For concreteness, let us consider the well known {\em prophet inequality} problem --- 
perhaps the most paradigmatic problem within Bayesian online settings. 

\paragraph{Prophet Inequality  \cite{krengel1977semiamarts,krengel1978semiamarts,samuel1984comparison}.} 
In a prophet inequality setting, there are $n$ boxes, each box $t=1,\ldots,n$ contains a value $v_t$ drawn from a known probability distribution $F_t$. 
The boxes arrive online. Upon the arrival of box $t$, its value $v_t$ is revealed, and the online algorithm needs to decide immediately and irrevocably whether to accept it, in which case the game ends with a reward of $v_t$, or to skip it, in which case the reward $v_t$ is lost forever and the game proceeds to the next box.
The prophet inequality problem captures many real-life scenarios where a decision maker inevitably makes decisions based on partial information regarding the future.
It has become central in the study of market design due to its close connection to mechanism design and posted price mechanisms~\cite{hajiaghayi2007automated}.

In its standard and most well-studied variant, the objective is to maximize the expected accepted value (corresponding to social welfare maximization in the market scenario above). 
Hence, 
the competitive ratio is the worst-case ratio (over all possible distribution sequences) between the expected value accepted by the algorithm, and the expected maximum value. 


The problem has been first studied under the assumption that the arrival order is known by the algorithm, where the optimal algorithm uses a sequence of $n$ different thresholds, computed by a backward induction process.
A celebrated result states that this {\em order-aware} algorithm achieves an expected value at least $1/2$ of the prophet \cite{krengel1977semiamarts,krengel1978semiamarts}, and this result is tight with respect to the prophet benchmark; that is, no online algorithm can obtain a better competitive ratio.

Quite remarkably, it has been later shown that the same competitive ratio of $1/2$ can be obtained by a single-threshold algorithm; namely, an algorithm that determines a single fixed threshold, and accepts the first value that exceeds it. Moreover, this algorithm is {\em order-unaware}, namely, in order to apply it, no information about the arrival order is needed~\cite{samuel1984comparison,KleinbergW19}.
Order-unawareness is a desirable property of online algorithms, as in many cases, no information about the arrival order is provided, and even if provided, order-unaware algorithms are robust against unexpected changes in the arrival order that may occur.
 


A natural question arises: what is the ratio between the performance of {\em order-unaware} and {\em order-aware} algorithms?
This essentially suggests a natural and more realistic benchmark for Bayesian online problems --- the best order-aware online algorithm. 
Such an algorithm still makes decisions online, but has an informational advantage over order-unaware algorithms, as it knows the arrival order in advance and can use this information to possibly make better decisions.
This leads us to the definition of a new performance measure for Bayesian online settings --- the {\em order-competitive ratio}.

\paragraph{Order-competitive ratio.}
We introduce the {\em order-competitive ratio}, defined as the worst-case ratio (over all distribution sequences) between 
the performance of the best {\em order-unaware} algorithm and the best {\em order-aware} algorithm)\footnote{Note that this new measure uses the same benchmark as the one in \cite{DBLP:conf/sigecom/PapadimitriouPS21,DBLP:conf/wine/NiazadehSS18}, but unlike \cite{DBLP:conf/sigecom/PapadimitriouPS21}, our algorithms are order-unaware, thus enjoy the robustness of order-unaware algorithms.}.
Thus, the order-competitive ratio {\bf quantifies the loss that is incurred by Bayesian online algorithms due to unknown arrival order}.

Recent years have seen a growing interest in the effect of the arrival order on the performance of Bayesian online algorithms.
Indeed, various arrival models have been studied, ranging from adversarial order~\cite{krengel1977semiamarts,krengel1978semiamarts,samuel1984comparison}, to random arrival order (a variant known as prophet secretary, as it combines the Bayesian assumption of prophet with the random order of secretary)~\cite{esfandiari2017prophet,azar2018prophet,ehsani2018prophet,CorreaSZ21}, all the way to ``free order", where the algorithm may dictate the arrival order~\cite{beyhaghi2018improved,AgrawalSZ20,PengT22}.
It is quite surprising that despite this surge of interest, the loss incurred due to lack of information about the arrival order (or put differently, the power obtained by knowledge of the arrival order) has not been studied.

To see that knowledge of the arrival order may be useful, consider the following example.

\begin{example} (see Figure~\ref{fig:example})
Suppose there are three boxes. Two boxes have deterministic values $\sqrt{2}$ and $1$, respectively, and one box has value $1/\epsilon$ with probability $\epsilon$ (and $0$ otherwise). Suppose further that the objective is to maximize the expected accepted value.   Suppose the first observed value is $\sqrt{2}$, and an order-unaware algorithm $\alg$ needs to decide immediately whether to accept it. If $\alg$ accepts it, the next arriving value is the random one, followed by $1$ (Figure~\ref{fig:example} (a)); else, the order flips, and $1$ arrives next, followed by the random box (Figure~\ref{fig:example} (b)). 
Consider now the best online algorithm that {\em knows} the order. In the former case (where $\alg$ accepted $\sqrt{2}$), it rejects $\sqrt{2}$, and  accepts the second value iff it is positive, gaining an expected value of $\sim 2$. The order-competitive ratio is then $1/\sqrt{2}$.
In the latter case (where $\alg$ rejected $\sqrt{2}$), it accepts $\sqrt{2}$, while $\alg$ gets an expected value of $1$, leading again to an order-competitive ratio of $1/\sqrt{2}$. 
\label{ex:gap}
\end{example}

 \begin{figure}
	\centering
	\includegraphics[width=0.9\textwidth]
	{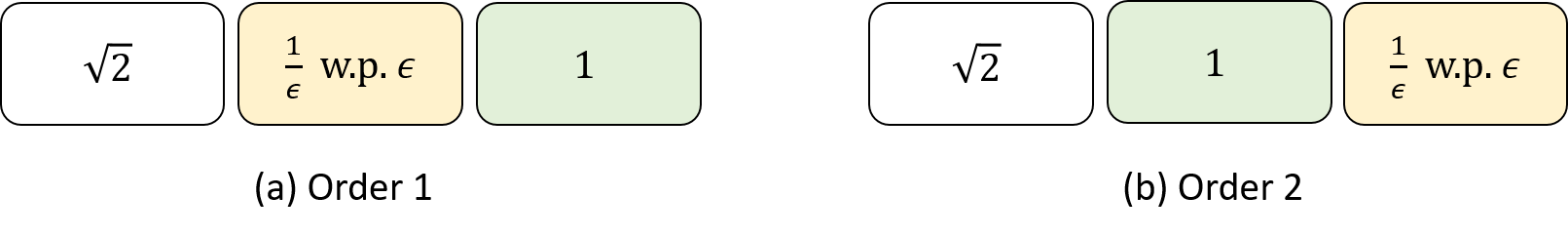}
	\caption{An example showing an upper bound of $1/\sqrt{2}$ on the order-competitive ratio of deterministic algorithms.
		\label{fig:example}}
\end{figure}

Example~\ref{ex:gap} shows that the order-competitive ratio cannot be better than $1/\sqrt{2} \approx 0.707$ (for deterministic algorithms). 
Clearly, it is at least $1/2$, as the order-competitive ratio is always weakly greater than the classic competitive ratio. 
It follows that the order-competitive ratio lies somewhere in the interval $[\frac{1}{2}, \frac{1}{\sqrt{2}}]$. 
As we show in this paper, the tight answer is the inverse of the golden ratio. 

\paragraph{Objective functions.}
The prophet inequality problem has been studied with respect to two well-motivated objective functions, namely (i) maximizing the expected value (hereafter, max-expectation objective); and (ii) maximizing the probability to catch the maximum value (hereafter, max-probability objective).

The max-expectation objective has been the subject of our discussion thus far; the max-probability objective has been first studied for the special case of i.i.d. random variables \citep{Gilbert1966}, and more recently in its general form \citep{esfandiari}, both with respect to the classical competitive ratio measure. 
This objective resembles the famous secretary problem \citep{ferguson1989solved} (but with values drawn from known distributions, and under an adversarial order, where the secretary problem considers adversarial values arriving in a random order). 
The prophet, who can see into the boxes, can trivially accept the maximum value with probability $1$, and so the competitive ratio of an algorithm is precisely its probability of catching the maximum value. 
Esfandiari et al.~\cite{esfandiari} devised an order-unaware single-threshold algorithm that gives a tight (up to lower-order terms) competitive ratio of $1/e$. 

\paragraph{Adaptive vs. single-threshold algorithms.}
The discussion above suggests that, when evaluated according to the competitive ratio measure, for both objective functions, single-threshold algorithms are optimal (providing tight $1/2$ and $1/e$ competitive ratios, respectively). 
Thus, the competitive ratio measure fails to differentiate between single-threshold and adaptive algorithms for Bayesian online settings. 
And yet, intuition suggests that adjusting one's decisions to the observed input may be quite useful.
This intuition is nicely captured by the order-competitive ratio.
Indeed, as we show in Section~\ref{sec:results}, the order-competitive ratio of adaptive algorithms is strictly better than the order-competitive ratio of single-threshold algorithms.
In this respect, the classical competitive ratio may provide  misleading algorithmic advice. 

\paragraph{Single-threshold algorithms.}
Single-threshold algorithms are particularly appealing, due to their simplicity and robustness.  
They also correspond to a particularly appealing policy in a natural market scenario captured by the prophet inequality problem, described as follows:
Consider a single-item sale, in a market with $n$ buyers, each with value $v_t$ drawn from a known distribution $F_t$.
The agents arrive to the market in an online fashion, revealing their values upon arrival. 
The seller needs to determine a selling policy. 
In this scenario, a single-threshold algorithm corresponds to a fixed price that is posted from the outset, where the first buyer whose value exceeds the price buys the item.

Interestingly, the single-threshold algorithms that obtain the $1/2$ and $1/e$ competitive ratios for the two objectives, respectively, are order-unaware.
However, single-threshold algorithms may be order-aware. 
In this paper we study the order-competitive ratio with respect to this class of algorithms as well. 
That is, we compare the performance of the best {\em order-unaware} single-threshold algorithm to that of the best {\em order-aware} single-threshold algorithm.


\subsection{Our Results}  \label{sec:results}

Our results are summarized in Table~\ref{tab:results}.
We study the order-competitive ratio with respect to two objective functions, namely, (i) maximizing the expected value (left column), and (ii) maximizing the probability to obtain the maximum value (right column); and with respect to two families of algorithms, namely, (i) adaptive algorithms (top row), and (ii) single-threshold algorithms (bottom row). 
We provide tight bounds for all four combinations, with respect to deterministic algorithms.

\begin{table}
\centering
\begin{tabular}{|cc|cc|}
\hline
\multicolumn{2}{|c|}{}                                     & \multicolumn{2}{c|}{\textbf{Objective function}}                                                \\ \cline{3-4} 
\multicolumn{2}{|c|}{}                                                                                      & \multicolumn{1}{c|}{\textbf{max-expectation}}      & \textbf{max-probability} \\ \hline
\multicolumn{1}{|c|}{{\textbf{Algorithm type}}} & \textbf{Adaptive}               & \multicolumn{1}{c|}{$\frac{1}{\phi}   \sim 0.618$} & $\sim   0.806$                              \\ \cline{2-4} 
\multicolumn{1}{|c|}{}                                               & \textbf{Single-threshold} & \multicolumn{1}{c|}{$\frac{1}{\phi}   \sim 0.618$} & $\sim   0.646$                              \\ \hline
\end{tabular}
\caption{Our results:
tight (with respect to deterministic algorithms) order-competitive ratios. 
Note that the competitive ratios in the two rows are defined with respect to different benchmarks (best adaptive and best single-threshold, respectively).}
\label{tab:results}
\end{table}

For example, the top left cell in the table shows that the order-competitive ratio with respect to the max-expectation objective function and adaptive deterministic algorithms is $1/\phi \approx 0.618$ (where $\phi$ is the golden ratio).
Namely, there exists a deterministic online {\em order-unaware} algorithm that obtains an expected value of at least $1/\phi \approx 0.618$ of the best online {\em order-aware} algorithm. Moreover, this is tight (with respect to deterministic algorithms).

This result is a significant improvement over the competitive ratio of $1/2$, as measured against the standard prophet benchmark.
As observed by \cite{DBLP:conf/wine/NiazadehSS18}, this result cannot be obtained by a single-threshold algorithm. 
Indeed, our algorithm is more complex and is inevitably adaptive. 
The other results in the table should be interpreted analogously, with respect to the corresponding objective functions (max-expectation or max-probability) and families of algorithms (adaptive or single-threshold).

Surprisingly, the same tight bound of $1/\phi \approx 0.618$ applies also for the family of single-threshold algorithms (bottom left cell), comparing the maximum expected value of the best deterministic online order-unaware single-threshold algorithm to that of the best online order-aware single-threshold algorithm. To the best of our understanding, the two results are unrelated. Indeed, despite obtaining the exact same ratio, they are derived using different analysis and techniques.

For the max-probability objective, we give a tight bound of $\sim 0.806$ for adaptive algorithms (top right cell), and a tight bound of $\sim 0.646$ for the family of single-threshold algorithms (bottom right cell).
The former result is a huge improvement over the competitive ratio of $1/e$ \citep{esfandiari}, with respect to the standard prophet benchmark (which is trivially 1, as the prophet can always select the box with the maximum value).
Furthermore, the upper bound of $0.646$ on the order-competitive ratio for the single-threshold algorithms implies that the best adaptive algorithm with order-competitive ratio $0.806$ cannot be single-threshold as opposed to the
optimal algorithm of~\citep{esfandiari} for the prophet benchmark. 


Recall that for the standard competitive ratio measure, the max-expectation objective yields better results than the max-probability objective. 
Indeed, the competitive ratio with respect to max-expectation is tightly $1/2$, while it is tightly $1/e$ with respect to max-probability. 
Interestingly, for the order-competitive ratio the order flips; namely, the max-probability objective exhibits better order-competitive ratios than the max-expectation objective. 

\vspace{0.1in}
\paragraph{Our Techniques.}
The algorithms for all of our settings use novel policies that depart from standard analysis of prophet inequalities. 
The derivation of these algorithms has been obtained by simultaneously considering policies that provide good bounds and analyzing their corresponding worst-case instances, and can be viewed as an application of the primal-dual approach. 
For simplicity of presentation, however, in our positive results we avoid explicit references to the worst-case instances, and focus instead on directly obtaining the essential inequalities. 
All worst-case instances turn out to consist of sequences of Bernoulli random variables with values $\{0,  a_i\}$, with vanishingly small probabilities of $a_i$ (in some cases we could simplify those instances to have only a constant number of boxes).
On the other hand, each of the four settings we consider requires a unique set of ideas, as we briefly discuss below.

First, our unrestricted algorithms for both max-expectation and max-probability objectives utilize \emph{adaptive thresholds}. Adaptive thresholds have been also used for multi-choice prophet inequality \citep{kleinberg2012matroid}. However, the analysis is fundamentally different since in our case the value of the benchmark is unknown, and may change depending on the arrival order. Part of the challenge is to gradually learn the benchmark and adjust the algorithm accordingly. This is in contrast to the prophet benchmark which is known from the outset.

Our approach for the max-expectation objective includes two novel ideas. First, we combine two different static threshold policies. Namely, at each step of the algorithm we take the \emph{higher of two} thresholds, which are variants of known thresholds for the classic prophet inequality.  
The first threshold is $\tau_1=\frac{1}{\phi}\Ex[\vals]{\max_i v_i}$; the second threshold is $\tau_2$ such that  $\Ex[\vals]{\left(\max_i v_i- \phi\tau_2\right)^+}=\tau_2$.
While the latter threshold is less known, a variant of it has already appeared in \citep{samuel1984comparison}.
The first threshold ($\tau_1$) is an estimation of the expected reward under the best order (scaled by $1/\phi$). In some cases, this threshold is too low. 
To this end, we consider also a second threshold ($\tau_2$), which is a scaled  estimation of the expected reward under the worst order.
Second, since the benchmark is unknown, we use adaptive variants of the corresponding thresholds. Namely, at each stage we define the thresholds with respect to the remaining boxes only.
Our algorithm essentially balances the two thresholds, in an adaptive way, to reflect the information we've gained about the (unknown) benchmark.
We are unaware of prior work on prophet inequalities 
that uses the better of two strategies at every individual step of the algorithm. 

For the max-probability objective, our analysis proceeds by
comparing the behavior of order-aware and order-unaware algorithms box-by-box, similar to the case of max-expectation. 
On the other hand, in the worst-case instance there is a deterministic box that comes first and the algorithm needs to commit whether to accept it or not.
If the order-unaware algorithm decides to select it, then the other boxes arrive in the optimal order (decreasing values). Otherwise, the boxes arrive in the worst order (increasing values). The hard instance is constructed by balancing these two cases.



In the case of single-threshold algorithms, each one of the order-aware and order-unaware algorithms can be described by a single real number (the corresponding threshold), leading to an analysis that involves 2 parameters.  By fixing these two thresholds, we obtain a relatively simple class of worst-case Bernoulli instances that minimizes the performance of order-unaware algorithms while maximizing the performance of order-aware ones.
(We do so separately for each one of our objective functions).
This process essentially induces a max-min optimization problem on the threshold values.
The most technically-involved part of our analysis is the solution of the induced max-min optimization problems.  

\subsection{Related Work}
\label{sec:related}

\paragraph{Different arrival models.}

Besides the adversarial order, random arrival order, and free order settings that we discussed above, another recent study related to the arrival order in prophet inequality settings has shown that for any arrival order $\pi$, the better of $\pi$ and the reverse order of $\pi$ achieves a competitive ratio of at least the inverse of the golden ratio, namely $1/\phi \approx 0.618$  \citep{arsenis2021constrained}. To the best of our understanding, while our first main result obtains the exact same ratio, the two results are unrelated.

\paragraph{Alternative benchmarks.}
Considering alternative benchmarks to the ``prophet'' (i.e., optimal offline) benchmark has attracted a lot of interest recently \citep{kessel2021stationary,DBLP:conf/wine/NiazadehSS18,DBLP:conf/sigecom/PapadimitriouPS21}.

For example, \cite{DBLP:conf/wine/NiazadehSS18} quantify the loss due to single-threshold algorithms by the worst-case ratio between the best single-threshold algorithm and the best adaptive online algorithm (single-threshold or not), both under a known order, and show that the $1/2$ ratio is tight.

As another example, \citet{DBLP:conf/sigecom/PapadimitriouPS21} consider the problem of online matching in bipartite graphs. This problem is known to admit a $1/2$-competitive algorithm with respect to the prophet benchmark \citep{FeldmanGL15}. 
But \citet{DBLP:conf/sigecom/PapadimitriouPS21} propose to study 
the ratio between the optimal polynomial (order-aware) algorithm and the optimal computationally-unconstrained  (order-aware) algorithm, and show that this ratio exceeds $1/2$.
(Note that this question makes sense in the matching variant, where the optimal online algorithm for matching is computationally hard even for known order.)
The ratio is then improved to $0.526$ by \citet{SaberiW21} and to $1-1/e$ by \citet{https://doi.org/10.48550/arxiv.2206.01270}.


\paragraph{Beyond single-choice settings.}
A related line of work, initiated by \citet{kennedy1985optimal,kennedy1987prophet,kertz1986comparison},
extends the single choice optimal stopping problem to  multiple-choice settings. More recent work extended it to additional combinatorial settings, including matroids~\citep{KleinbergW19,azar2014prophet}, 
polymatroids~\citep{dutting2015polymatroid}), matching~\citep{GravinW19,EzraFGT20}, combinatorial auctions \citep{FeldmanGL15,DuttingFKL20,dutting2020log}, and general downward closed feasibility constrains~\citep{rubinstein2016beyond}.  

\paragraph{Limited information models.}
Prophet inequality problems have been also studied under limited information about the underlying distributions, where the emphasis is on the sample complexity of the problem ~\citep{azar2014prophet,correa2019prophet,correa2020two,ezra2018prophets,rubinstein2019optimal,kaplan2020competitive,correa2021secretary,dutting2021secretaries,caramanis2022single,kaplan2022online}.

%% file: prelim.tex


Consider a setting with $n$ boxes. 
Every box $t$ contains some value $v_t$ drawn from an underlying independent distribution $F_t$. 
The underlying distributions are known from the outset, but the values are revealed sequentially in an online fashion.
For convenience of notations, we assume that the boxes arrive in an order $1,2,\ldots,n$. 
I.e., at stage $t$, we observe the realized value $v_t = \theta_t$, where $v_t \sim F_t$, {and the identity of the arriving box.} 
It will be clear from the context whether the order of arrival is assumed to be known. 
In any case, the identity of the arriving box is known.
We denote by $\dists=\prod_t F_t$ the (product) distribution of the value profile
$\vals=(v_1,\ldots,v_n)$.
Upon the arrival of value $v_t = \theta_t$, the algorithm needs to decide, immediately and irrevocably, whether to accept the box.

We consider two different objectives:
(i) maximizing the expected value of the accepted box, and (ii) maximizing the probability of catching the box with the largest value.


An online algorithm is said to be {\em order-aware} if it knows the arrival order of the boxes from the outset, and is said to be {\em order-unaware} if it doesn't. 

Our goal is to measure the performance of order-unaware algorithms against that of the best online order-aware algorithm.
Our order-unaware algorithm will be denoted by $\alg$, and the best order-aware algorithm by $\opt$. 

Given an arrival order $\pi$ and a value profile $\vals$, we denote by $\alg(\vals,\pi)$ the value accepted by $\alg$ under values $\vals$ and arrival order $\pi$, and by $\opt(\vals,\pi)$ the value accepted by $\opt$ under $\vals,\pi$.


We denote by $\alg(\pi)$ the performance of $\alg$. 
For the first objective, it is the expected accepted value, i.e., 
$
\alg(\pi)=\Ex[\vals\sim\dists]{\alg(\vals,\pi)}
$.
For the second objective, it is the probability of catching the maximum value, i.e., 
$
\alg(\pi)=\prob[\vals\sim\dists]{\alg(\vals,\pi)=\max_i v_i}.
$
When studying the objective of catching the maximum value, {we assume that} the maximum value is unique with probability $1$ (as in the case where, e.g., the supports of the distributions are disjoint, or where the distributions are atomless). 

The order-competitive ratio of an order-unaware algorithm $\alg$ measures the loss in performance due to unknown order. 
It is defined as the worst-case ratio of the performance of $\alg$ and the performance of $\opt$, over all arrival orders. 

\begin{definition}
The order-competitive ratio of an order-unaware algorithm $\alg$ is 
$$
\Gamma(\alg) = \min_{\pi}\frac{\alg(\pi)}{\opt(\pi)}
$$
\end{definition}

%% file: prelim2.tex
We also consider an important restricted family of algorithms, namely single-threshold algorithms.  
A single-threshold algorithm $\alg_{\tau}$, parameterized by a single number\footnote{More generally, the threshold $\tau$ can be chosen at random. But we only consider deterministic single-threshold algorithms.} $\tau\ge 0$, stops at the first box $i$ with $v_i\ge\tau$.
We denote by $\alg_{\tau}(\pi)$ the algorithm's expected value for the arrival order $\pi$. I.e., 
$
\alg_{\tau}(\pi)=\Ex[\vals\sim\dists]{\alg_{\tau}(\vals,\pi)}
$.

When studying the objective of catching the maximum value, we assume that the distributions are atomless, i.e., every cumulative distribution function is continuous. As standard 
(see, e.g., \cite{esfandiari}), we extend the definition of single-threshold algorithms for discrete distributions by allowing the algorithm to randomize when the realized value $v_i$ equals to the threshold $\tau$.

We note that the threshold that gives the best performance may depend on the arrival order $\pi$ and consequently there is a gap between order-unaware and order-aware single-threshold algorithms.
We define the order-competitive ratio for a single-threshold algorithm $\alg_{\tau}$, with respect to the class of single-threshold algorithms, as the worst-case (over all arrival orders $\pi$) ratio between $\alg_{\tau}(\pi)$ and $\alg_{\tau'}(\pi)$, for any single threshold $\tau'$ (where the threshold $\tau'$ may depend on $\pi$, but $\tau$ is unaware of $\pi$).

\begin{definition}
The order-competitive ratio of a single-threshold algorithm $\alg_{\tau}$, with respect to the class of single threshold algorithms, is 
$$
\Gamma_{ST}(\alg_{\tau}) = \inf_{\pi,\tau'}\frac{\alg_{\tau}(\pi)}{\alg_{\tau'}(\pi)}.
$$
\end{definition}

Recall that if one compares the performance of a single-threshold algorithm to the best online algorithm, then a competitive ratio of $1/2$ is tight, even if the arrival order is known \citep{DBLP:conf/wine/NiazadehSS18}.

%% file: max_expectation.tex

In this section we study the objective function of maximizing the expected value. Section~\ref{sec:max-expected-general} gives the order-competitive ratio with respect to adaptive algorithms, and Section~\ref{sec:expectation-st} gives the order-competitive ratio with respect to single-threshold algorithms.

\subsection{Adaptive Algorithms}
\label{sec:max-expected-general}
Our main result in this section is a deterministic order-unaware algorithm that obtains a tight order-competitive ratio of the inverse of the golden ratio (i.e., $\frac{1}{\phi} \approx 0.618$) with respect to the objective of maximizing the expected value.
	
\paragraph{An order-unaware algorithm.}
For the convenience of notations, we assume the boxes arrive in a specific order from $1,2,\ldots,n$. I.e., at stage $t$, we observe $v_t = \theta_t$ where $v_t \sim F_t$. It will be clear from the description of our algorithm that it is order-unaware.
We define the following series of random variables
\[
\text{Prophet in the future:}\quad y_t \eqdef \max_{s > t} v_s.
\]
At each step $t\in[n]$, our algorithm will use the larger of the following two thresholds
\begin{equation}
\label{eq:beta_def}
\alpha_t \eqdef \frac{1}{\phi} \cdot \Ex{y_t}
\quad\quad
\beta_t\eqdef x \text{ satisfying }  \Ex{\left(y_t - \phi \cdot x\right)^+} = x,
\end{equation}
where $\phi = \frac{\sqrt{5}+1}{2}$ is the golden ratio. 
When $t=n$, $\alpha_t=\beta_t=0$. Note that the equation defining $\beta_t$ has a unique solution, since its LHS is a strictly decreasing continuous function in $x$ that starts from a non-negative number when $x=0$ and goes to $0$ for $x=\infty$, and the RHS is a strictly increasing continuous function in $x$ that starts from $0$ when $x=0$ and goes to $\infty$ when $x=\infty$.
Our algorithm $\alg$ stops at box $t$ if and only if the realized value $\theta_t$ of $v_t$ exceeds the threshold $\tau_t$, defined as follows 
\begin{equation}
\label{eq:tau-def}
\tau_t\eqdef\max (\alpha_{t}, \beta_{t}).
\end{equation}

Note that $\alg$ is order-unaware, since it does not need to know the arrival order of the remaining boxes in order to calculate $\alpha_t,\beta_t$. 
Denote the expected value obtained by this algorithm as $\alg$ and the one obtained by the best order-aware algorithm as $\opt$. 

Our main theorem in this section is the following.

\begin{theorem}
	\label{th:pouo-det}
	  For every arrival order $\pi$, $\alg(\pi) \ge \frac{1}{\phi}\opt(\pi)$, where $\phi$ is the golden ratio.
\end{theorem}


\begin{proof} 
Fix an order $\pi$.
Let $\alg_t$ denote the expected value of $\alg$ when run only on the boxes from $t$ to $n$. We first establish the following useful bound on the performance of $\alg_t$ relative to the thresholds $\alpha_t,\beta_t$.
\begin{lemma}
	\label{lem:key-det}
It holds that $\alg_{t+1} \ge \beta_t \ge \frac{1}{\phi^2} \Ex{y_t} = \frac{1}{\phi} \alpha_t$ for any $t\in[n-1]$.
\end{lemma}
\begin{proof}
According to the definition of $\beta_t$, $\beta_t = \Ex{\left(y_t - \phi \cdot \beta_t\right)^+} \ge \Ex{y_t} - \phi \cdot \beta_t$. 
Hence, $\beta_t \ge \frac{1}{\phi + 1} \Ex{y_t} = \frac{1}{\phi^2} \Ex{y_t}$, which concludes the proof of the second inequality.

We next prove the first inequality by induction on the number of remaining boxes. We use $t=n$ as the base case of our induction, which is satisfied since $\alpha_n=\beta_n=0$. Assume that $\alg_{t+1}\ge \beta_t$. We need to prove that $\alg_{t}\ge \beta_{t-1}$. First, observe that
\begin{align}
\label{eq:algt_beta}
\alg_t & = \Ex[v_t]{\idr{v_t \geq \tau_t} \cdot v_t + \idr{v_t < \tau_t} \cdot \alg_{t+1}}\nonumber \\
& \ge \Ex[v_t]{\idr{v_t \geq  \tau_t} \cdot v_t + \idr{v_t < \tau_t} \cdot \beta_t}
\ge \beta_t,
\end{align}
where the first inequality holds by the induction hypothesis, and to obtain the second inequality we observe that $\tau_t\ge\beta_t$.
Consider the difference function $f(x)\eqdef\Ex{\left(y_{t-1} - \phi \cdot x\right)^+} -x$ (see Equation \eqref{eq:beta_def}). $f$ is strictly decreasing, with $f(x)=0$ for $x=\beta_{t-1}$, by definition of $\beta_{t-1}$. Therefore, to prove that $\alg_{t}\ge\beta_{t-1}$ it is sufficient to prove that $f(\alg_t)\le 0$. We have 
\begin{eqnarray}
    f(\alg_t) & = &\Ex[y_{t-1}]{\left(y_{t-1}-\phi\cdot\alg_t\right)^+}-\alg_t  
    \le \Ex[y_{t-1}]{\left(y_{t-1}-\phi\cdot\beta_t\right)^+}-\alg_t \nonumber\\ & =& 
    \Ex[y_t,v_t]{\left(\max(y_t,v_t)-\phi\cdot\beta_t\right)^+}-\alg_t  \nonumber \\
    & \le &  \Ex[y_t,v_t]{\left(\max(y_t,\tau_t)-\phi\cdot\beta_t\right)^{+} + (v_t-\tau_t)^{+}}-\alg_t\nonumber \\
    &\le &  \Ex[y_t]{\left(\max(y_t,\tau_t)-\phi\cdot\beta_t\right)^{+}}
    + \Ex[v_t]{(v_t-\tau_t)^{+}} - \Ex[v_t]{\idr{v_t \geq \tau_t} \cdot v_t + \idr{v_t < \tau_t} \cdot \beta_t}\nonumber\\
    & = &  \Ex[y_t]{\left(\max(y_t,\tau_t)-\phi\cdot\beta_t\right)^{+}} -
    \Ex[v_t]{\idr{v_t\geq  \tau_t} \cdot \tau_t + \idr{v_t < \tau_t} \cdot \beta_t} \nonumber\\
    & \le & \Ex[y_t]{\left(\max(y_t,\tau_t)-\phi\cdot\beta_t\right)^{+}} - \beta_t, \label{eq:algt_longderivation}
\end{eqnarray}
where the first inequality follows since $\alg_t \geq \beta_t$ by Equation \eqref{eq:algt_beta}; the second inequality holds since $(\max(a,b)-c)^+\le (\max(a,d)-c)^+ + (b-d)^+$ for any $a,b,c,d\in\mathbbm{R}$; to get the third inequality, we use the first part of Equation \eqref{eq:algt_beta}; the last inequality follows not only in expectation over $v_t$ but for any fixed value $\theta_t$ of $v_t$, as $\tau_t=\max(\alpha_t,\beta_t)\ge\beta_t$. 
Furthermore,
\begin{multline*}
f(\alg_t) \stackrel{\eqref{eq:algt_longderivation}}{\le}  \Ex[y_t]{\left(\max(y_t,\tau_t)-\phi\cdot\beta_t\right)^{+}}-\beta_t  =  
\Ex[y_t]{\max\left\{\left(y_t-\phi\cdot\beta_t\right)^{+},\left(\tau_t-\phi\cdot\beta_t\right)^{+}\right\}}-\beta_t\\
\le
\Ex[y_t]{\left(y_t-\phi\cdot\beta_t\right)^{+} + \left(\tau_t-\phi\cdot\beta_t\right)^{+}}-\beta_t
 =
\left(\Ex[y_t]{\left(y_t-\phi\cdot\beta_t\right)^{+}}-\beta_t\right) + {\left(\tau_t-\phi\cdot\beta_t\right)^{+}}
 = 0,
\end{multline*}
where the first inequality is precisely \eqref{eq:algt_longderivation}; the second inequality follows by observing that $\max(a,b)\le a+b$ for any $a,b\in\\mathbbm{R}_+$; the last equality follows by observing that both terms equal $0$. The first term ($\Ex[y_t]{\left(y_t-\phi\cdot\beta_t\right)^{+}}-\beta_t$) equals $0$ by the definition of $\beta_t$, and the second term (${\left(\tau_t-\phi\cdot\beta_t\right)^{+}}$) equals $0$ by recalling that $\tau_t = \max (\alpha_{t}, \beta_{t})$ and by the fact that 
$\beta_t \ge \frac{1}{\phi} \alpha_t$ proved above.
Thus $f(\alg_t)\le 0$ and $\alg_t\ge \beta_{t-1}$, which concludes the proof of the induction step. 
\end{proof}

We are now ready to prove Theorem~\ref{th:pouo-det}.
We prove the statement of the theorem 
by induction on the total number of boxes $n$. For the base case ($n=1$), $\alpha_1=\beta_1=0$ and $\alg_n=\opt_n$.
Suppose that the statement of the theorem holds for any $n-1$ boxes. 
Let $\opt_t$ 
denote the expected value of the optimal order-aware algorithm 
on the boxes $t, \ldots, n$. 
We shall prove the induction step that $\alg_1=\alg\ge\frac{1}{\phi}\opt=\frac{1}{\phi}\opt_{1}$ for the case of $n$ boxes. By the induction hypothesis we have $\alg_{2} \geq \frac{1}{\phi}\opt_{2}$. 
We consider four cases based on the realized value $\theta_1$ of the first box. We denote by $\alg(\theta_1)$ and $\opt(\theta_1)$ the respective expected values of our algorithm and the optimal order-aware algorithm, given that the value in the first box is $v_1=\theta_1$. We show that $\alg(\theta_1)\ge\frac{1}{\phi}\opt(\theta_1)$ for any $\theta_1$.
	\begin{description}
		\item[Case 1] Both $\alg$ and $\opt$ stop and take value $\theta_1$. Then, $\alg_1(\theta_1)=\opt_1(\theta_1)=\theta_1$.
		\item[Case 2] $\alg$ takes value $\theta_1$ but $\opt$ doesn't. Then, $\alg(\theta_1) \ge \alpha_1 = \frac{1}{\phi}\Ex{y_1} \ge \frac{1}{\phi}\opt_2= \frac{1}{\phi}\opt_1(\theta_1)$, {where the second inequality is} since $\opt_2$ cannot do better than the prophet on boxes $t\in\{2,\ldots,n\}$.
		\item[Case 3] $\opt$ takes $\theta_1$, but $\alg$ doesn't. It holds that $\alg_1(\theta_1) = \alg_{2} \ge \max(\beta_1,\frac{1}{\phi}\alpha_1) \ge \frac{1}{\phi}\max(\beta_1,\alpha_1) \ge \frac{1}{\phi}\opt_1$, where the first inequality follows by  Lemma~\ref{lem:key-det}, and the last inequality holds since $\alg$ rejected $\theta_1<\tau_1$, whereas $\opt$ selected it (thus $\opt(\theta_1)=\theta_1 < \tau_1 = \max(\beta_1,\alpha_1)$).
		\item [Case 4] Neither $\alg$ nor $\opt$ takes $\theta_1$. Then, $\alg(\theta_1)=\alg_{2}$ and $\opt(\theta_1)=\opt_{2}$, and the claim holds by the induction hypothesis.
	\end{description}
Therefore, $\alg=\Ex[v_1]{\alg(v_1)}\ge\Ex[v_1]{\frac{\opt(v_1)}{\phi}}
=\frac{1}{\phi}\opt$. This concludes the proof.
\end{proof}

We next show that the above bound is tight, namely that no order-unaware deterministic algorithm may achieve an order-competitive ratio better than the golden ratio $\frac{1}{\phi}= \frac{2}{\sqrt{5}+1}$. 

\begin{theorem}
	For the objective of maximizing the expected value, no deterministic order-unaware algorithm achieves a better order-competitive ratio than $\frac{1}{\phi}$ in the worst case. 
\end{theorem}
\begin{proof}
    Consider an instance that consists of a set of boxes with deterministic values $\phi, \phi-\eps, \phi-2\eps, \ldots, 1$, and a single random box $HV$ (we call it a high variance box) with value $1/\eps$ realized with probability $\eps$ and value $0$ realized otherwise. 
	Let $\alg$ be any given order-unaware deterministic algorithm. 
	Let $\pi$ be an arrival order where the deterministic boxes arrive first, in decreasing order: $\phi, \phi-\eps, \phi-2\eps, \ldots, 1$, followed by the $HV$ box. 
	
	\begin{description}
	\item[Case 1] $\alg$ accepts some deterministic value $x > 1$. 
	Consider now another arrival order $\pi_x$ which is the same as $\pi$ up to the deterministic $x$ box, but with the $HV$ box arriving immediately after it, and followed by the remaining deterministic boxes $x-\epsilon, \ldots, 1$ in any order.
	Then, $\alg$ achieves value $x$ (we slightly abuse notations, and denote it $\alg(\pi_x)=x$), whereas $\opt$ for $\pi_x$ achieves $\opt(\pi_x)=  \eps \cdot 1/\eps +(1-\eps)(x-\eps)$, by waiting for the $HV$ box and taking it when its realized value is $1/\eps$ (otherwise $\opt$ takes $x-\eps$). 
	As $\eps$ goes to 0, $\alg/\opt$ goes to
	$\frac{x}{1+x}\le \frac{\phi}{1+\phi}= \frac{1}{\phi}$, where the inequality follows since $x\le \phi$.
	
	\item[Case 2] $\alg$ waits for the last deterministic item or the $HV$ box. Then, $\alg(\pi)=1$, whereas $\opt(\pi)$ selects the first item (with value $\phi$), leading to an order-competitive ratio of $\frac{1}{\phi}$.
	\end{description}
\end{proof}

%% file: max_expectation_st.tex
In this section we study the order-competitive ratio with respect to single-threshold algorithms. 
Interestingly, the exact same bound of the inverse of the golden ratio holds also with respect to single-threshold algorithms. 
That is, we provide a single-threshold deterministic order-unaware algorithm whose ratio with respect to the optimal single-threshold order-aware algorithm is at least the inverse of the golden ratio (i.e., $\frac{1}{\phi} \approx 0.618$), and this is tight.

\paragraph{An order-unaware algorithm.}
Let $\tau$ be the (unique) value satisfying

\begin{equation}
    \Ex{\vmax \cdot \idr{\vmax \ge \tau}} \geq \phi \cdot \tau \geq \Ex{\vmax \cdot \idr{\vmax > \tau}}, \label{eq:tau}
\end{equation}
where $\vmax=\max_{i}v_i$ is the random variable equal to the maximum value of $v_1,\ldots,v_n$.

\begin{theorem}\label{thm:single}
For every   arrival order $\pi$, and for every threshold $\tau'$, the performance of the single-threshold algorithm $\alg_\tau$ is at least $\frac{1}{\phi}$-competitive against the single-threshold algorithm $\alg_{\tau'}$. I.e., for every $\pi,\tau'$, it holds that $$\alg_\tau(\pi) \geq \frac{1}{\phi} \alg_{\tau'}(\pi).$$
\end{theorem}

\begin{proof}
Fixing the arrival order $\pi$, and the threshold $\tau'$,
we first show that for the threshold $\tau''=\min(\tau',\alg_{\tau'}(\pi))$ it holds that $\alg_{\tau''}(\pi) \geq \alg_{\tau'}(\pi)$, i.e., $\tau''$ is at least as good as  $\tau'$.
\begin{lemma}
$\alg_{\tau''}(\pi) \geq \alg_{\tau'}(\pi)$.
\end{lemma}
\begin{proof}
If $\tau' \leq \alg_{\tau'}(\pi)$, then $\tau''=\tau'$ and $\alg_{\tau''}(\pi) = \alg_{\tau'}(\pi)$, in which case we are done.

Assume $\tau'>\alg_{\tau'}(\pi)$.
Let $\alg_{\tau'}^t(\pi) = \Ex{\alg_{\tau'}(\pi) \mid v_1,\ldots, v_{t-1} <\tau' }$ denote the expected value obtained by running the threshold $\tau'$ given that boxes
$1, \ldots, t-1$ were below $\tau'$. We first establish the following useful bound on the performance of $\alg_{\tau'}^t(\pi)$.

\begin{claim}
\label{cl:optt}
If $\tau' > \alg_{\tau'}(\pi)$, then for every $t=1,\ldots,n$, it holds that $\alg_{\tau'}(\pi) \geq \alg_{\tau'}^t(\pi)$.
\end{claim}
\begin{proof}
It holds that 
\begin{eqnarray*}
\alg_{\tau'}(\pi) & = & \Ex{\alg_{\tau'}\mid \exists i<t, v_i \geq \tau'
}\cdot \Prx{ \exists i<t, v_i \geq \tau'}) +  (1-\Prx{ \exists i<t, v_i \geq \tau'}) \cdot \alg_{\tau'}^t(\pi) \\ &\geq &
\tau' \cdot \Prx{ \exists i<t, v_i \geq \tau'}) +  (1-\Prx{ \exists i<t, v_i \geq \tau'}) \cdot \alg_{\tau'}^t(\pi) \\ &\geq&
\alg_{\tau'}(\pi) \cdot \Prx{ \exists i<t, v_i \geq \tau'}) +  (1-\Prx{ \exists i<t, v_i \geq \tau'}) \cdot \alg_{\tau'}^t(\pi) ,
\end{eqnarray*}
where the first inequality follows by $\Ex{\alg_{\tau' }\mid \exists i<t, v_i \geq \tau'
}\geq \tau'$, since $\alg_{\tau'}$ must stop before $t$ and accept value $v\ge\tau'$ under the condition $\exists i<t, v_i\ge\tau'$. 
The second inequality follows by the assumption that $\tau' > \alg_{\tau'}(\pi)$.
Since $\tau'>\alg_{\tau'}(\pi)$, we get that the probability that we stop by time $t$ is less than $1$, thus, by rearranging and dividing by $1-\Prx{ \exists i<t, v_i \geq \tau'}$ which is strictly positive, we get
$\alg_{\tau'}(\pi) \geq \alg_{\tau'}^t(\pi)$, as desired.
\end{proof}

Next we show that $\alg_{\tau''}(\pi) \geq \alg_{\tau'}(\pi) $. Consider a fixed valuation profile $\vals$ and the following two situations.

First, if $\alg_{\tau''}$ and $\alg_{\tau'}$ accept the same box in $\vals$, then clearly $\alg_{\tau''}$ and $\alg_{\tau'}$ obtain the same value.

Second, let $t$ be the first index where $\alg_{\tau'}$ and $\alg_{\tau''}$ diverge.
Since $\tau''=\alg_{\tau'}(\pi)<\tau'$, it means that $\alg_{\tau''}$ accepts box $t$, while $\alg_{\tau'}$ doesn't.
Let $\alg_{\tau''}(v_1,\ldots,v_t)$ and $\alg_{\tau'}(v_1,\ldots,v_t)$ denote the expected value of $\alg_{\tau''}$ and $\alg_{\tau'}$, respectively, with fixed values $v_1, \ldots, v_t$, and where expectation is taken over $v_{t+1}, \ldots, v_n$.
Since $\alg_{\tau''}$ stops at box $t$, it holds that $\alg_{\tau''}(v_1,\ldots,v_t) \geq \tau''$.
Since $\alg_{\tau'}$ doesn't stop at box $t$, it holds that $\alg_{\tau'}(v_1,\ldots,v_t) =\alg_{\tau'}^{t+1}(\pi)  \leq \alg_{\tau'}(\pi)  = \tau''$, where the inequality follows by Claim~\ref{cl:optt}.
It follows that $\alg_{\tau''}(\pi) \geq \alg_{\tau'}(\pi)$, as desired.
\end{proof}

We now distinguish between two cases.

\paragraph{Case 1:} $\tau'' \geq \tau$. We assume that $\tau''>\tau$ (when $\tau''=\tau$ we are done).
Then
$$
\alg_{\tau''}(\pi) \le \Ex{ \vmax \cdot \idr{\vmax \ge \tau''}} \le \Ex{ \vmax \cdot \idr{\vmax > \tau}} \leq \phi \cdot \tau.
$$
Thus, it remains to show that $\alg_{\tau}(\pi) \ge \tau$.
We next use that $\alg_{\tau''}(\pi) \geq \alg_{\tau'}(\pi)  \geq \tau''$. 

\begin{eqnarray}
\label{eq:tautau}
\tau'' & \leq & \alg_{\tau''}(\pi) \nonumber\\
& = & \sum_{t=1}^{n} \Prx{v_1,\ldots,v_{t-1} < \tau''}\cdot \Ex{(v_t-\tau'')^+} + 
(1-\Prx{v_1,\ldots,v_{n} < \tau''})\cdot \tau''\nonumber\\
& = & \sum_{t=1}^{n} \prod_{i=1}^{t-1} \Prx{v_i < \tau''}\cdot \Ex{(v_t-\tau'')^+} + 
\left(1-\prod_{t=1}^{n}\Prx{v_t < \tau''}\right)\cdot \tau''.
\end{eqnarray}
Rearranging, we get
$$
\sum_{t=1}^{n} \prod_{i=1}^{t-1} \Prx{v_i < \tau''}\cdot \Ex{(v_t-\tau'')^+} \geq \tau''\cdot \prod_{t=1}^{n}\Prx{v_t < \tau''},
$$
which is equivalent to 
$$
\sum_{t=1}^{n} \Ex{(v_t-\tau'')^+} \prod_{i=t}^{n} \frac{1}{\Prx{v_i < \tau''}} \geq \tau''.
$$
Since $\tau < \tau''$ (by Case 1), the last inequality must hold also when replacing $\tau''$ by $\tau$. Indeed, the LHS increases while the RHS decreases. That is:
\begin{equation}
\label{eq:one_over_prod_prob}
\sum_{t=1}^{n} \Ex{(v_t-\tau)^+} \prod_{i=t}^{n} \frac{1}{\Prx{v_i < \tau}} \geq \tau.
\end{equation}
Similar to the derivation of \eqref{eq:one_over_prod_prob} from \eqref{eq:tautau} with $\tau''$ replaced by the threshold $\tau$ and $\alg_{\tau''}(\pi)$ by $\alg_{\tau}(\pi)$, inequality \eqref{eq:one_over_prod_prob} implies that 
$\alg_{\tau}(\pi) \geq \tau$, which concludes the proof of the first case.

\paragraph{Case 2:} $\tau'' < \tau$. Then
\begin{equation}
\alg_{\tau''}(\pi)  \le \alg_{\tau}(\pi) + \Prx{\vmax < \tau} \cdot \tau, \label{eq:algtau2}
\end{equation}
since for any valuation profile $\vec{v}$ we have either $\alg_{\tau}(\pi)\ge\alg_{\tau''}(\pi)$ whenever $\alg_{\tau}$ stops (i.e., $\vmax\ge\tau$), or $\alg_{\tau''}(\pi)<\tau$ whenever $\vmax<\tau$ (i.e., $\alg_{\tau}$ does not stop). Furthermore,
\begin{eqnarray}
\alg_{\tau}(\pi) & \ge &\Prx{\vmax \ge \tau} \cdot \tau + \Prx{\vmax < \tau } \cdot \Ex{(\vmax - \tau)^+ } \nonumber\\
& \ge & \Prx{\vmax < \tau } \cdot \left( \Ex{(\vmax - \tau)^+ } + \Prx{\vmax \ge \tau} \cdot \tau \right) \nonumber\\
& \geq &\Prx{\vmax < \tau } \cdot \phi \cdot \tau, \label{eq:algtau}
\end{eqnarray}
where the first inequality is by the standard decomposition of the expected value of a threshold algorithm $\alg_{\tau}$ into revenue and surplus, and the last inequality is equivalent to Equation~\eqref{eq:tau}.
Combining Equations~\eqref{eq:algtau2} and \eqref{eq:algtau}, it follows that 
$$
\frac{\alg_{\tau''}(\pi)}{\alg_{\tau}(\pi)} \leq 1+\frac{1}{\phi} = \phi.
$$
\end{proof}

\begin{remark}
We note that the performance guarantee of Theorem~\ref{thm:single} applies to single threshold algorithms that select the first value that is strictly greater than $T$ (rather than  at least $T$). This is since such an algorithm with threshold $T$ can be interpreted as   $\lim_{T'\rightarrow T^+}\alg_{T'}$.
\end{remark}

We next show that the above bound is tight, namely that no single-threshold order-unaware deterministic algorithm may achieve an order-competitive ratio (with respect to the optimal single-threshold algorithm) better than the inverse of the golden ratio $\frac{1}{\phi}= \frac{2}{\sqrt{5}+1}$.

\begin{theorem}
	For the objective of maximizing the expected value, no deterministic single-threshold order-unaware algorithm achieves a better order-competitive ratio (with respect to the optimal single-threshold algorithm) than $\frac{1}{\phi}$ in the worst case. 
\end{theorem}

\begin{proof}
    Consider an instance that consists of two boxes. One box with a deterministic value 1, the other with value $\phi/\eps$ with probability $\eps$ (and $0$ otherwise).
If the threshold is at most $1$, then the deterministic box arrives first, and we get $\alg=1$ and $\opt=\phi$. 
If the threshold is greater than $1$, then the deterministic box arrives last, and we get $\alg=\phi$, while $\opt=\phi+1$.
In any case the ratio is $1/\phi$.
\end{proof}

%% file: max_prob.tex
In this section we study the objective function of maximizing the probability to catch the maximum value. Section~\ref{sec:max-prob-general} gives the order-competitive ratio with respect to adaptive algorithms, and Section~\ref{sec:prob-st} gives the order-competitive ratio with respect to single-threshold algorithms.

\subsection{Adaptive Algorithms}
\label{sec:max-prob-general}

Our main result in this section is a deterministic order-unaware algorithm that obtains an order-competitive ratio of $\ln \frac{1}{\lambda} \approx 0.806$, where $\lambda\approx 0.4464$ is the unique solution to $\frac{x}{1-x} = \ln \frac{1}{x}$, with respect to the objective of maximizing the probability to catch the maximum value. 

To prove this result, we consider a slightly more general game: let there be an extra number $\theta$ given in advance, and our objective is to maximize the probability of catching the box with the largest value that exceeds $\theta$. If all boxes have values less than $\theta$, no algorithm wins. 

From now on, we shall work on this variant of the problem. Observe that the original problem is a special case where $\theta=0$.

\paragraph{An order-unaware algorithm.} At round $t \in [n]$, let $v_s = \theta_s$ be the realized values for each $s \le t$. Let $\theta_0=\theta$. We accept the current box $v_t=\theta_t$ if it satisfies the following condition: 
\[
\theta_t = \max_{0 \le s \le t} \theta_s \quad \text{and} \quad
\Prx[v_{t+1},\ldots,v_n]{\max_{t+1 \le s \le n} v_s < \theta_t} \ge \lambda~,
\]
where $\lambda$ is the unique solution to $\frac{x}{1-x} = \ln \frac{1}{x}$.

Note that our algorithm is order-unaware since calculating the probability $\Prx{\max_{t+1 \le s \le n} v_s < \theta_t}$ requires no information on the order of remaining boxes. Indeed, it is the probability that all remaining boxes have value less than $\theta_t$.
We use $\alg(\pi)$ to denote the probability of catching the maximum value that exceeds $\theta$ by our order-unaware algorithm and $\opt(\pi)$ to denote the winning probability of the best order-aware algorithm when the actual arrival order is $\pi$. 
\begin{theorem}
\label{thm:catching_max}
	For every arrival order $\pi$, the probability of catching the maximum value of the algorithm satisfies
	$\alg(\pi) \ge \ln \frac{1}{\lambda} \cdot \opt(\pi)$.
\end{theorem}
\begin{proof}
For simplicity, we omit $\pi$ and write $\alg$ and $\opt$ instead of $\alg(\pi)$ and $\opt(\pi)$, respectively.
We prove the statement by induction on the number of boxes.
The base case when $n=1$ is trivial, since both our algorithm and the optimal algorithm would accept the first box with value $v_1$ if and only if $v_1 > \theta_0$.
Suppose the statement is correct for $n-1$ boxes and consider the case for $n$ boxes. We shall prove that for any realized value of $v_1$ the winning probability of our algorithm is at least $\ln \frac{1}{\lambda}$ times the winning probability of the optimal algorithm.

Consider the four cases depending on the behavior of our algorithm and the optimal algorithm on the realization of the first box $v_1$. 
\begin{description}
\item[Case 1] both $\alg$ and $\opt$ accept. In this case, $\alg = \opt$.
\item[Case 2] both $\alg$ and $\opt$ reject. In this case, we update the current maximum to $\theta_1=\max (\theta_0, v_1)$ and apply the induction hypothesis to conclude the proof. We remark that this is the place where we need the generalized version of the problem.
\item[Case 3] $\alg$ accepts and $\opt$ rejects. Then we have 
\[
\alg = \Prx{\max_{s\ge 2} v_s < \theta_1} \ge \lambda \quad \text{and} \quad \opt \le \Prx{\max_{s\ge 2} v_s \ge \theta_1} \le 1-\lambda,
\]
due to the second condition of our algorithm.
Therefore, $\alg \ge \frac{\lambda}{1-\lambda} \cdot \opt$.
\item[Case 4] $\alg$ rejects and $\opt$ accepts. In this case, we must have $\theta_1 \ge \theta_0$ and $\Pr[\max_{s \ge 2} v_s < \theta_1] < \lambda$. Let $2 \le t \le n$ be the index such that 
\begin{equation}
\label{eqn:switching_point}
\Prx{\max_{s>t} v_s < \theta_1} \ge \lambda \quad \text{and} \quad \Prx{\max_{s\ge t} v_s < \theta_1} < \lambda.
\end{equation}
We define $p_s = \Prx{v_s \ge \theta_1}$ for all $s \ge 2$.
Since the optimal algorithm takes the first box with realized value of $\theta_1$, its winning probability is
\begin{equation}
\label{eqn:opt}
\opt = \Prx{\max_{s \ge 2} v_s < \theta_1} = \prod_{2\le s\le n} (1-p_s).
\end{equation}
Next, we analyze our algorithm by studying the following events $A_s$ for $t \le s \le n$:
\[
A_s \eqdef \left\{ v_s \ge \theta_1 \text{ and } v_k < \theta_1, \forall k\ne s \right\}.
\]
\begin{claim}
\label{cl:alg-win}
Our algorithm wins if $A_s$ happens for any $t \le s \le n$.
\end{claim}
\begin{proof}
It suffices to show that our algorithm accepts box $s$. Indeed, $\alg$ does not stop before box $s$ since all other boxes have values smaller than $\theta_1$, violating the first condition of our algorithm. The second stopping condition of $\alg$ is satisfied for box $s$ as
\[
\Prx{\max_{k > s} v_k < \theta_s} \ge \Prx{\max_{k > s} v_k < \theta_1} \ge \Prx{\max_{k > t} v_k < \theta_1} \ge \lambda,
\]
where the last inequality is by the first inequality of \eqref{eqn:switching_point}. 
\end{proof}
 
Finally, we conclude the proof of the case:
\begin{multline*}
\alg \ge \Prx{\bigcup_{t \le s \le n} A_s} = \sum_{t \le s \le n} \left( p_s \cdot \prod_{\substack{2 \le k \le n \\ k\ne s}} (1-p_k) \right) = \prod_{2 \le k \le n} (1-p_k) \cdot \left( \sum_{t\le s \le n} \frac{p_s}{1-p_s}\right) \\
\ge \prod_{2 \le k \le n} (1-p_k) \cdot \left( \sum_{t\le s \le n} \ln \frac{1}{1-p_s}\right) = \prod_{2 \le k \le n} (1-p_k) \cdot \ln \left(\frac{1}{\prod_{t\le s \le n}(1-p_s)}\right) \ge \opt \cdot \ln\frac{1}{\lambda}.
\end{multline*}
The first inequality follows by Claim~\ref{cl:alg-win}.
The second inequality holds since $\frac{x}{1-x} \ge \ln \frac{1}{1-x}$ for any $x \in [0,1]$. The last inequality follows from Equations~\eqref{eqn:switching_point} and \eqref{eqn:opt}.
\end{description}
This finishes the proof of the inductive step.
\end{proof}

We next show that the above result is tight, namely that no order-unaware  deterministic algorithm may achieve an order-competitive ratio better than $\ln \frac{1}{\lambda} \approx 0.806$, where $\lambda$ is the unique solution to $\frac{x}{1-x} = \ln \frac{1}{x}$.
\begin{theorem}
	For the objective of catching the maximum value, no deterministic order-unaware algorithm achieves a better order-competitive ratio than $\ln \frac{1}{\lambda} \approx 0.806$.
\end{theorem}

\begin{proof}
Consider an instance with $n+1$ boxes. One of the boxes has a deterministic value of $1/2$. Among the remaining $n$ boxes, the $i$-th box (for $i=1,\ldots,n)$ has value $v_i = i$ with probability $\eps$, and value $0$ otherwise, where $n\to\infty$, and $\eps$ is a small number that satisfies $(1-\eps)^n = \lambda$.
Let the deterministic box come first.
\begin{itemize}
\item Suppose the algorithm accepts the deterministic box. Then the algorithm wins when all the remaining boxes have value $0$, i.e. $\alg = (1-\eps)^n = \lambda$. Then, the remaining boxes arrive in decreasing order, i.e., $n,n-1,\ldots,1$. 
The optimal algorithm that knows the arrival order would reject the deterministic box and accept the first box that has non-zero value. This algorithm wins when at least one of the randomized boxes has non-zero value, i.e. $\opt = 1-(1-\eps)^n = 1-\lambda$. This gives us an order-competitive ratio of $\frac{\lambda}{1-\lambda}$.
\item Suppose {the} algorithm rejects the deterministic box. Then, the remaining boxes arrive in 
increasing order, i.e., $1,2,\ldots,n$. 
We study the best online algorithm that knows the order afterwards. It is straightforward to check that 1) the optimal algorithm only accepts non-zero boxes; 2) if the optimal algorithm accepts the $i$-th box when $v_i \ne 0$, it should also accept the $j$-th box for all $j \ge i$ when $v_j \ne 0$. I.e., such algorithm can be described by a single parameter $s \in [n]$ and it would simply accept the first non-zero box after the $s$-th box.
Its winning probability is
\begin{eqnarray*}
\alg  & = &\Prx{\text{exactly one randomized box of } \{s,s+1,\ldots,n\} \text{ has non-zero value}} \\
& = & (n-s+1) \cdot \eps \cdot  (1-\eps)^{n-s} \le n\cdot  \eps \cdot  (1-\eps)^{n-1}.
\end{eqnarray*}
The last expression approaches $\lambda \ln \frac{1}{\lambda}$ as $n\to\infty$.
On the other hand, the optimal order-aware algorithm $\opt$ would simply accept the deterministic box and 
win with probability $\lambda$ (if all remaining boxes have value $0$).  
Again, this leads to an order-competitive ratio of $\frac{\lambda \ln \frac{1}{\lambda}}{\lambda} = \ln \frac{1}{\lambda}$.
\end{itemize}
\end{proof}



%% file: max_prob_st.tex


In this section, we provide a single-threshold deterministic order-unaware algorithm whose ratio with respect to the optimal single-threshold order-aware algorithm (with respect to the objective of maximizing the probability to catch the maximum value) is at least $\mu = \frac{\ln (1/\rho)}{\ln (1/\rho) + 1} \approx 0.646$, where $\rho \in [0,1] (\approx 0.1609)$ is the solution to the following equation. 
\begin{equation} \label{eq:rho}
\min_{\rho' \in (\rho, 1)} \left( \frac{\rho}{\rho'} + \frac{\rho}{1-\rho'} \cdot \ln\frac{\rho'}{\rho} \right) = \frac{\ln( 1/\rho)}{\ln(1/\rho) + 1}
\end{equation}
Moreover, the ratio is tight.
{Since the LHS (resp. the RHS) of Equation~\eqref{eq:rho} is strictly increasing (resp., decreasing) in $\rho$, it admits a unique solution.}

\paragraph{An order-unaware algorithm.}
Let $\tau$ be the (unique) threshold value satisfying
\begin{equation}
    \Prx{\max_{i} v_i < \tau} = \rho . \label{eq:prob_tau}
\end{equation}
The existence of $\tau$ is guaranteed by the atomless assumption. We remark that this is the only use of the assumption throughout the proof.

\begin{theorem}\label{thm:single-maxprob}
For every   arrival order $\pi$, and for every threshold $\tau'$, the performance of the single-threshold algorithm $\alg_\tau$ is at least $\mu$-competitive against the single-threshold algorithm $\alg_{\tau'}$. I.e., for every $\pi,\tau'$, it holds that $$\alg_\tau(\pi) \geq \mu \cdot \alg_{\tau'}(\pi).$$
\end{theorem}

\begin{proof}
We fix the arrival order $\pi$ and the threshold $\tau'$. We omit $\pi$ and write $\alg_\tau$ and $\alg_{\tau'}$ instead of $\alg_\tau(\pi)$ and $\alg_{\tau'}(\pi)$, respectively, to simplify notations. Let $\rho' = \Prx{\max_i v_i < \tau'}$. 
We consider the two following cases.
\paragraph{Case 1:} $\tau' \geq \tau$. We assume that $\tau'>\tau$ (when $\tau'=\tau$ the statement is obviously true). Then
\begin{equation}
\label{eq:bound_alg_tau_prime}
1-\rho'= \Prx{\max_{1 \le i \le n} v_i \ge \tau'} \ge
\alg_{\tau'}  = \sum_{1 \le i \le n} \Prx{\max_{j<i} v_j < \tau'} \cdot \Prx{v_i \ge \tau', \max_{j>i} v_j < v_i}.
\end{equation}
We write $\alg_{\tau}$ in a similar way to $\alg_{\tau'}$.
\begin{align*}
\alg_{\tau}  &= \sum_{1 \le i \le n} \Prx{\max_{j<i} v_j < \tau} \cdot \Prx{v_i \ge \tau, \max_{j>i} v_j < v_i} 
 = A + B, ~~\text{where}\\
& A\eqdef\sum_{1 \le i \le n} \Prx{\max_{j<i} v_j < \tau}  
 \cdot \Prx{v_i \ge \tau', \max_{j>i} v_j < v_i}\\
& B\eqdef \sum_{1 \le i \le n} \Prx{\max_{j<i} v_j < \tau} \cdot \Prx{v_i \in [\tau, \tau'), \max_{j>i} v_j < v_i} 
 \end{align*}
 Now, we have $\Prx{\max_{j<i} v_j < \tau}\ge
 \Prx{\max_{j<i} v_j < \tau'}\cdot\frac{\Prx{\max_{j} v_j<\tau}}{\Prx{\max_{j} v_j<\tau'}}$  in the $A$ term, as $\Prx{\max_j v_j<\tau'}=\Prx{\max_{j<i} v_j<\tau'}\cdot\Prx{\max_{j\ge i} v_j<\tau'}\geq \Prx{\max_{j<i} v_j<\tau'}\cdot\Prx{\max_{j\ge i} v_j<\tau}$. Thus,
\[
A\ge \frac{\Prx{\max_{j} v_j<\tau}}{\Prx{\max_{j} v_j<\tau'}}\cdot\sum_{1 \le i \le n} \Prx{\max_{j<i} v_j < \tau'}  
 \cdot \Prx{v_i \ge \tau', \max_{j>i} v_j < v_i}=\frac{\rho}{\rho'}\alg_{\tau'}.  
\]
In the $B$ term, $\Prx{v_i \in [\tau, \tau'), \max_{j>i} v_j < v_i}\ge \Prx{v_i \in [\tau, \tau')}\cdot\Prx{\max_{j>i} v_j < \tau}$. Thus,
\begin{eqnarray*}
B & \ge & \sum_{1 \le i \le n} \Prx{\max_{j<i} v_j < \tau} \cdot \Prx{v_i \in [\tau, \tau')}\cdot\Prx{\max_{j>i} v_j < \tau} \\ 
& = & \Prx{\max_j v_j<\tau}\cdot
\sum_{1\le i\le n}
\frac{\Prx{v_i < \tau'}-\Prx{v_i < \tau}}{ \Prx{v_i < \tau}}\ge \rho\cdot\sum_i\ln \frac{\Prx{v_i < \tau'}}{\Prx{v_i < \tau}}=\rho \cdot \ln \frac{\rho'}{\rho},
\end{eqnarray*}
where in the second inequality we use the fact that $x\ge \ln (1+x)$ and the definition of $\rho$. By combining these bounds on $A$ and $B$ together we get
\[
\alg_{\tau}\ge\frac{\rho}{\rho'} \cdot \alg_{\tau'} + \rho \cdot \ln\left(\frac{\rho'}{\rho}\right)\ge
\alg_{\tau'}\cdot\left(\frac{\rho}{\rho'}+
\frac{\rho}{1-\rho'} \cdot \ln\left(\frac{\rho'}{\rho}\right)
\right)\ge\mu\cdot\alg_{\tau'},
\]
where the second inequality follows by \eqref{eq:bound_alg_tau_prime}.
\paragraph{Case 2:} $\tau' < \tau$. We write a similar 
decomposition for $\alg_{\tau'}$ as we did for $\alg_{\tau}$ in case 1. 

\begin{multline*}
\alg_{\tau'} = \sum_{1 \le i \le n} \Prx{\max_{j<i} v_j < \tau'} \cdot \Prx{v_i \ge \tau, \max_{j>i} v_j < v_i}+ \\
   \sum_{1 \le i \le n} \Prx{\max_{j<i} v_j < \tau'}\cdot\Prx{v_i \in [\tau', \tau), \max_{j>i} v_j < v_i}\\ 
 \le \sum_{1 \le i \le n} \Prx{\max_{j<i} v_j < \tau} \cdot \Prx{v_i \ge \tau, \max_{j>i} v_j < v_i} + \Prx{\max_{i} v_i < \tau} 
 = \alg_\tau + \rho
\end{multline*}


We also have the following lower bound on $\alg_{\tau}$. 
\begin{align*}
\alg_{\tau} & \ge \sum_{1 \le i \le n} \Prx{v_i \ge \tau, \max_{j \ne i} v_j <\tau } = \rho \sum_{i} \frac{\Prx{v_i \ge \tau}}{\Prx{v_i < \tau}}  
\ge \rho\sum_{i}\ln\left(\frac{1}{\Prx{v_i<\tau}}\right)
= \rho \cdot \ln \frac{1}{\rho},
\end{align*}
where the second inequality follows by the fact that $\frac{x}{1-x}>\ln(\frac{1}{1-x})$ for $x\in[0,1)$. Therefore,
\begin{eqnarray*}
\alg_{\tau} & \ge& \alg_{\tau'}\cdot\frac{\alg_{\tau}}{\alg_{\tau}+\rho}=\alg_{\tau'}\cdot\left(1-\frac{\rho}{\alg_{\tau}+\rho}\right)\\
& \ge &\alg_{\tau'}\cdot\left(1-\frac{\rho}{\rho\ln(1/\rho)+\rho}\right)
=\alg_{\tau'}\cdot\frac{\ln(1/\rho)}{\ln(1/\rho)+1}=\mu\cdot\alg_{\tau'}
\end{eqnarray*}
\end{proof}


We next show that the guarantee of Theorem~\ref{thm:single-maxprob} is tight.
\begin{theorem}
	For the max-probability objective, no single-threshold order-unaware algorithm achieves a better order-competitive ratio (with respect to the optimal single-threshold algorithm) than $\mu\approx 0.646$ in the worst case. 
\end{theorem}

\begin{proof}
Consider an instance that consists of $n \rightarrow \infty$ boxes. 
For $i=1,\ldots,n$, the value of box $i$ is $i$ with probability $\eps$, and $0$ otherwise\footnote{{This instance can be transformed into an instance with atomless distributions by adding a small noise of less than $1$ to all values. The same proof applies to the noisy instance.}}, for $\eps \rightarrow 0$ such that $\eps \gg \frac{1}{n}$. 
Consider an arbitrary single-threshold algorithm, and let $T$ be its threshold. 
Without loss of generality, we assume that $T$ is an integer and the single-threshold algorithm always accepts the box $T$ when $v_T=T$, since $\Prx{v_T = T} = \eps \to 0$.
We distinguish between two cases, depending on the value of $T$.

\paragraph{Case 1:} The threshold $T$ satisfies $(1-\eps)^{n-T} \leq \rho$, where $\rho$ is the solution to Equation~\eqref{eq:rho}. 
Then, let $\rho' = (1-\eps)^{n-T}$, let
$\rho^*=\arg\min_{\rho'' \in (\rho',1)} \left( \frac{\rho'}{\rho''} + \frac{\rho'}{1-\rho''} \cdot \ln\frac{\rho''}{\rho'} \right) $, and let $T^*= \max\{T' \in [n] \mid (1-\eps)^{n-T'} \leq \rho^*\}$.
Since $\rho^*\geq \rho'$, it holds that $T^*\geq T$.

Consider an arrival order $\pi$, where
  \[
    \pi_i =
    \begin{cases*}
      \mbox{Box } i& if $1 \leq i <  T^*$ \\
      \mbox{Box } n+T^*-i & if
      $i \geq  T^*$
    \end{cases*}
  \]

Let $X_1$ (respectively, $X_2$, $X_3$) denote the random variable indicating the number of non-zero realized values below $T$ (respectively, between $T$ and $T^*$, and above $T^*$).
Thus, 
\begin{eqnarray} \label{eq:rhoprime}
\alg_T(\pi) & = &  \Prx{X_2=1 \wedge X_3=0} + \Prx{X_2=0 \wedge X_3 >0} \nonumber \\
 & \approx &(T^*-T) \cdot \eps \cdot (1-\eps)^{T^*-T} \cdot (1-\eps)^{n-T^*} + (1-\eps)^{T^*-T} \cdot (1-(1-\eps)^{n-T^*}) \nonumber \\ 
 &  = &(T^*-T) \cdot \eps \cdot (1-\eps)^{n-T} +(1-\eps)^{T^*-T} - (1-\eps)^{n-T}  \nonumber \\
 & \approx & \eps \cdot \rho' \cdot \log_{1-\eps} \frac{\rho'}{\rho^*} + \frac{\rho'}{\rho^*} - \rho' \approx  \rho' \ln \frac{\rho^*}{\rho'} + \frac{\rho'}{\rho^*} - \rho' .
\end{eqnarray}
On the other hand, 
\begin{equation} \label{eq:rhostar}
\alg_{T^*}(\pi) = \Prx{X_3>0} \approx 1- (1-\eps)^{n-T^*} \approx  1-\rho^*. 
\end{equation}
Combining Equations~\eqref{eq:rhoprime} and \eqref{eq:rhostar}, we get that $$ \frac{\alg_T(\pi)}{\alg_{T^*}(\pi) } \approx \frac{\rho' \ln \frac{\rho^*}{\rho'} + \frac{\rho'}{\rho^*} - \rho'}{ 1-\rho^*} = \frac{\rho'}{\rho^*} + \frac{\rho'}{1-\rho^*} \ln \frac{\rho^*}{\rho'} \leq  \mu,$$
where the last inequality is by the definitions of $\mu,\rho^*$, since $\rho' \leq \rho$, {and since the LHS of Equation~\eqref{eq:rho} is an increasing function.}

\paragraph{Case 2:} The threshold $T$ satisfies $(1-\eps)^{n-T} > \rho$. 
Then, let $T^*= 1$.
Consider an arrival order $\pi$, where
  \[
    \pi_i =
    \begin{cases*}
      \mbox{Box } T+i & if $0 \leq i \leq   n-T$ \\
      \mbox{Box } n-i & if
      $i >  n-T$.
    \end{cases*}
  \]
Let $Y_1$ (respectively, $Y_2$) denote the random variable indicating the number of non-zero realized values of at least $T$ (respectively, below $T$).
Thus, 
\begin{eqnarray}\label{eq:alg4}
\alg_T(\pi) & = & \Prx{Y_1=1} \approx (n-T) \cdot \eps \cdot (1-\eps)^{n-T}  \nonumber \\ 
& = & \ln_{1-\eps} \rho' \cdot \eps \cdot \rho' \approx \frac{\rho'}{\ln(1/\rho')} .
\end{eqnarray}
On the other hand, 
\begin{eqnarray}\label{eq:opt4}
\alg_{T^*}(\pi) & = & \Prx{Y_1=1} + \Prx{Y_1=0 \wedge Y_2 >0} \nonumber \\
& \approx  & \frac{\rho'}{\ln(1/\rho')} + (1-\eps)^{n-T} = \frac{\rho'}{\ln(1/\rho')} + \rho'.
\end{eqnarray}
Combining Equations~\eqref{eq:alg4} and \eqref{eq:opt4}, we get that
$$ \frac{\alg_T(\pi)}{\alg_{T^*}(\pi) } \approx \frac{\frac{\rho'}{\ln(1/\rho')}}{\frac{\rho'}{\ln(1/\rho')} + \rho'} = \frac{\ln(1/\rho')}{\ln(1/\rho') + 1} \leq  \mu,$$
where the last inequality is by the definition of $\mu$, since $\rho'>\rho$, {and since the RHS of Equation~\eqref{eq:rho} is a decreasing function.}
\end{proof}

%% file: open_problems.tex
\section{Open Problems}

Our model and results suggest natural problems for future research.
\begin{enumerate}
\item 
Our bounds are tight with respect to deterministic algorithms. 
Can randomized algorithms provide better ratios? 
\item
Study the order-competitive ratio in combinatorial settings, where multiple elements can be accepted, subject to feasibility constraints.  
A clear candidate is matroid feasibility constraints, for which the competitive ratio of $1/2$ with respect to the prophet benchmark carries over \citep{kleinberg2012matroid}. 
Our bounds for the expected value objective carry over to simple matroid settings, such as partition matroids (where a single element is chosen from each part). 
What is the order-competitive ratio for general matroids?

\item
More generally, we believe that the order-competitive ratio is a meaningful measure, which captures the significance of knowing the arrival order in Bayesian online settings. It would be interesting to apply it to other Bayesian settings.
\item The algorithm proposed by \citep{DBLP:conf/sigecom/PapadimitriouPS21} for online bipartite matching is order-aware. Can the same result be obtained by an order-unaware algorithm?


\end{enumerate}